\documentclass[submission,copyright,creativecommons]{eptcs}
\providecommand{\event}{GandALF 2016} 
\usepackage{breakurl}             

\usepackage{amsthm}
\usepackage{makeidx}  
\usepackage{amsmath}
\usepackage{amssymb}

\usepackage{paralist}
\usepackage{tikz}
\usetikzlibrary{shapes,arrows,chains,positioning,
decorations,automata,backgrounds,petri,matrix,fit,decorations.markings,decorations.pathreplacing}
\usepackage{algorithm}
\usepackage{algorithmic}
\usepackage{stmaryrd}

\def\β{\beta}
\def\γ{\gamma}\def\Γ{\Gamma}
\def\φ{\varphi}
\def\ψ{\psi}
\def\τ{\tau}\def\Τ{T}
\def\ρ{\rho}
\def\λ{\lambda}\def\Λ{\Lambda}
\def\Σ{\Sigma}
\def\π{\pi}\def\Π{\Pi}
\def\ω{\omega}
\def\χ{\chi}\def\Χ{X}
\def\×{\times}
\def\⊆{\subseteq}
\def\|{\mid}
\def\→{\to}

\def\<{\langle}\def\>{\rangle}
\def\∵{\because}\def\∴{\therefore}

\def\bl{\texttt{\char32}}

\newcommand{\Reachable}{\mathop{\mathrm{Reachable}}\nolimits}
\newtheorem{theorem}{Theorem}[section]
\newtheorem{lemma}{Lemma}[section]
\newtheorem{proposition}{Proposition}[section]
\newtheorem{corollary}{Corollary}[section]
\theoremstyle{definition}
\newtheorem{definition}{Definition}[section]
\theoremstyle{remark}
\newtheorem*{remark}{Remark}
\theoremstyle{remark}
\newtheorem{example}{Example}[section]

\title{The Almost Equivalence by Asymptotic Probabilities for Regular Languages and Its Computational Complexities}
\author{Yoshiki Nakamura
\institute{Tokyo Institute of Technology\\ Tokyo, Japan}
\email{nakamura.y.ay@m.titech.ac.jp}
}
\def\titlerunning{The Almost Equivalence by Asymptotic Probabilities for Regular Languages}
\def\authorrunning{Yoshiki Nakamura}
\begin{document}
\maketitle

\begin{abstract}
We introduce \emph{$p$-equivalence} by asymptotic probabilities, which is a weak almost-equivalence based on zero-one laws in finite model theory.
In this paper, we consider the computational complexities of $p$-equivalence problems for regular languages and provide the following details.
First, we give an robustness of $p$-equivalence and a logical characterization for $p$-equivalence.
The characterization is useful to generate some algorithms for $p$-equivalence problems by coupling with standard results from descriptive complexity.
Second, we give the computational complexities for the $p$-equivalence problems by the logical characterization.
The computational complexities are the same as for the (fully) equivalence problems.
Finally, we apply the proofs for $p$-equivalence to some generalized equivalences.
\end{abstract}

\section{Introduction}
  The study of the equivalence problem of regular languages dates back to the beginning of formal language theory.
  This problem is a fundamental problem and regular languages have many applications (see e.g., \cite{BKM01}).
  Regular expressions (REG), nondeterministic finite state automaton (NFA), and deterministic finite state automaton (DFA) are normally used to represent regular languages.
  Both the equivalence problem for NFAs and REGs are known as PSPACE-complete \cite{meyer1973word} and the equivalence problem for DFAs is known as NL-complete \cite{jones1975space}.

  In recent years, some \emph{almost-equivalences} for regular languages were introduced.
  These equivalences are weaker than the (fully) equivalence.
  For example, two languages, $L_1$ and $L_2$, are \emph{$f$-equivalent} \cite{Badr:2008:HO:1428728.1428753,ITA:8238099} if their symmetric difference, $L_1 \vartriangle L_2$\footnote{$L_1 \vartriangle L_2 = (L_1 \setminus L_2) \cup (L_2 \setminus L_1)$}, is a finite set; and two languages, $L_1$ and $L_2$, are \emph{$E$-equivalent} \cite{DBLP:conf/dlt/HolzerJ12} if their symmetric difference, $L_1 \vartriangle L_2$, is a subset of $E$, where $E$ is a regular language.
  In \cite{DBLP:conf/dlt/HolzerJ12}, it is pointed out that both $f$-equivalence problems and $E$-equivalence problems for NFAs are PSPACE-complete; and both $f$-equivalence problems and $E$-equivalence problems for DFAs are NL-complete, where the regular language $E$ is given by a DFA $\mathcal{A}_E$ as an input.
  In this paper, we define another almost-equivalence (\emph{$p$-equivalence}).
  $p$-equivalence is defined as follows.
  Let $\mu_n(L)$ be
  $$\mu_n(L) = \frac{\text{the number of strings of length $n$ that are in $L$}}{\text{the number of strings of length $n$}}.$$
  That is, $\mu_n(L)$ is the probability that a randomly chosen string of length $n$ is in a language $L$.
  The \emph{asymptotic probability} of $L$, $\mu(L)$, is defined as $\mu(L) = \lim_{n \to \infty} \mu_n(L)$ if the limit exists.
  Then, we define that two languages, $L_1$ and $L_2$, are \emph{$p$-equivalent} if $\mu(L_1 \triangle L_2) = 0$.

  The definition is based on the asymptotic probabilities in finite model theory, which are defined as follows.
  Let $\mu_n(\Phi)$ be
  $$\mu_n(\Phi) = \frac{\text{the number of finite graphs with $n$ nodes that satisfy $\Phi$}}{\text{the number of finite graphs with $n$ nodes}}.$$
  That is, $\mu_n(\Phi)$ is the probability that a randomly chosen graph with $n$ nodes satisfies a first-order sentence $\Phi$.
  (Note that this definition can be extended to any finite $\sigma$-structures from finite graphs.)
  The asymptotic probability of $\Phi$, $\mu(\Phi)$, is defined as $\mu(\Phi) = \lim_{n \to \infty} \mu_n(\Phi)$ if the limit exists.
  Then, we define that $\Phi$ is \emph{almost surely valid} if $\mu(\Phi) = 1$.

  In finite model theory, the next two theorems are some interesting results in decidability between validity and ``almost surely'' validity.
  \begin{theorem}[Trakhtenbrot \cite{trakhtenbrot1950impossibility}]\label{thm : undecidable}
    For any vocabulary $\sigma$ with at least one binary relation symbol,
    it is \emph{undecidable} whether a first-order sentence $\Phi$ of vocabulary $\sigma$ is valid over finite $\sigma$-structures.
  \end{theorem}
  \begin{theorem}[see e.g., Corollary 12.11 \cite{libkin2013elements}]\label{thm : decidable}
    There is an algorithm that given as input a finite $\sigma$-structure and a first-order sentence $\Phi$ of vocabulary $\sigma$, decides whether $\Phi$ is almost surely valid.
  \end{theorem}
  Relative to finite $\sigma$-structures, Theorem \ref{thm : decidable} tells us that it is \emph{decidable} whether a sentence is almost surely valid, whereas Theorem \ref{thm : undecidable} tells us that it is \emph{undecidable} whether a sentence is valid.
  One of our main motivation to consider $p$-equivalence is as follows:
  Does there exist some differences in decidability or in computational complexity between equivalence and $p$-equivalence?
  
  (In this paper, however, in the class of regular languages, we prove that there is no differences in computational complexity between equivalence and p-equivalence, e.g., the p-equivalence problem for REGs is also PSPACE-complete.)
\subsection*{Our results and contributions.}
  In this paper, we give the computational complexities of the $p$-equivalence problems for regular languages.
  Moreover, we also give these complexities of some generalized equivalence problems.

  First, we give a simple characterization of $p$-equivalence,
  coupled with standard results from descriptive complexity \cite{immerman2012descriptive}, which is used to decide the p-equivalence problem for various representations of regular languages.

  Second, we prove the computational hardness for the $p$-equivalence problems by modifying the proofs of the computational hardness for (fully) equivalence problems.

  Finally, we give the computational complexities for equivalence problems for some generalized equivalences based on the proofs for the $p$-equivalence problems.
  These results give a robustness of equivalence problems for regular languages in terms of the computational complexities when the equivalence is generalized.
  \subsection*{Paper outline.}
    The remainder of this paper is organized as follows:
    Section \ref{section : pre} gives the necessary definitions and terminology for languages, automaton, and $p$-equivalence;
    Section \ref{section : fund} shows some fundamental results of $p$-equivalence;
    Section \ref{section : upperbound} describes the computational complexity upper bounds of both the $p$-equivalence problems and some generalized equivalence problems;
    Section \ref{section : lowerbound} describes the computational complexity lower bounds of both the $p$-equivalence problems and some generalized equivalence problems;
    Section \ref{section : zero-one law} remarks about the problem to decide whether a given regular language obeys zero-one law \cite{DBLP:journals/corr/Sinya15a} based on previous sections.
\section{Preliminaries}\label{section : pre}
    In this paper, we consider three well-known standard models for regular languages, \emph{regular expression} (REG), \emph{deterministic finite state automaton} (DFA), and \emph{nondeterministic finite state automaton} (NFA).

    Let $A$ be a finite alphabet and let $A^*$ [$A^n$] be the set of all strings [of length $n$] over $A$.
    \begin{description}
        \item[REG]
            The syntax for REG is defined as follows:
            $$\alpha := 0 \mid 1 \mid a \in A \mid \alpha_1 \cdot \alpha_2 \mid \alpha_1 \cup \alpha_2 \mid \alpha_1^*$$
            Then, $L(\alpha)$ (the language of REG $\alpha$) is inductively defined as follows:\\
            \begin{inparaenum}[(1)]
                \item $L(0) = \emptyset$;
                \item $L(1) = \{\epsilon\}$;
                \item $L(a) = \{a\}$;
                \item $L(\alpha_1 \cdot \alpha_2) = L(\alpha_1) \cdot L(\alpha_2)$;
                \item $L(\alpha_1 \cup \alpha_2) = L(\alpha_1) \cup L(\alpha_2)$; and
                \item $L(\alpha_1^*) = \bigcup_{n \ge 0} \overbrace{L(\alpha_1) \cdot \ldots \cdot L(\alpha_1)}^\text{$n$ times}$,
            \end{inparaenum}
            where the concatenation operation $\cdot$ is defined as $L(\alpha_1) \cdot L(\alpha_2) = \{s_1 s_2 \mid s_1 \in L(\alpha_1),s_2 \in L(\alpha_2) \}$.
            We may omit $\cdot$ (i.e., $\alpha_1 \alpha_2$ denotes $\alpha_1 \cdot \alpha_2$).
            $\epsilon$ denotes the empty string.
        \item[DFA]
            A DFA $\mathcal{A}$ is a 5-tuple $(Q,A,\delta,q^0,F)$, where
            \begin{inparaenum}[(1)]
              \item $Q$ is a finite set of states;
              \item $A$ is a finite alphabet;
              \item $\delta : Q \times A \to Q$ is a transition function;
              \item $q^0 \in Q$ is the initial state; and
              \item $F \subseteq Q$ is a set of acceptance states.
            \end{inparaenum}
            We inductively define $\delta(q,s)$ by using the definition of $\delta(q,a)$ as follows.
            If $s = \epsilon$, then $\delta(q,s) = q$.
            Otherwise (i.e, $s = a s'$), $\delta(q,s) = \delta(\delta(q,a),s')$.
        
            Then, $L(\mathcal{A}) = \{ s \in A^* \mid \delta(q^0,s) \in F \}$.
        \item[NFA]
            A NFA $\mathcal{A}$ is a 5-tuple $(Q,A,\delta,q^0,F)$, where
            \begin{inparaenum}[(1)]
              \item $Q$ is a finite set of states;
              \item $A$ is a finite alphabet;
              \item $\delta : Q \times A \to 2^Q$ is a transition function;
              \item $q^0 \in Q$ is the initial state; and
              \item $F \subseteq Q$ is a set of acceptance states.
            \end{inparaenum}
            Let $\delta(Q',a) = \bigcup_{q \in Q'} \delta(q,a)$, where $Q' \subseteq Q$ and we inductively define $\delta(Q',s)$ by using the definition of $\delta(Q',a)$ as follows.
            If $s = \epsilon$, then $\delta(Q',s) = Q'$.
            Otherwise (i.e, $s = a s'$), $\delta(Q',s) = \delta(\delta(Q',a),s')$.
        
            Then, $L(\mathcal{A}) = \{ s \in A^* \mid \exists q \in \delta(q^0,s). q \in F\}$.
    \end{description}

    $\Reachable(q,q')$ in DFA[NFA] means that there exists a string $s$ such that $\delta (q,s) = q'$[$q' \in \delta (\{q\},s)$].

  \subsection{The almost equivalence by asymptotic probabilities and the zero-one law for formal language theory}
    The zero-one law in finite model theory is a property which means ``almost surely true" or ``almost surely false" (see e.g., \cite[Section 12]{libkin2013elements}).
    In formal language theory, zero-one law is investigated by Sin'ya \cite{DBLP:journals/corr/Sinya15a} as follows;
    A language $L$ obeys zero-one law if almost all strings are in $L$ or almost all strings are not in $L$.
    In other words, a language $L$ obeys zero-one law if $L$ is ``almost empty" or ``almost full".
    Formally, ``almost empty" and ``almost full" are defined by asymptotic probabilities.
    Let $L$ be a language.
    We define
    $$\mu _n(L) = \frac{|\{s \in A^n \mid s \in L\}|}{|A^n|}$$
    That is, $\mu _n(L)$ is the probability that a string of $n$ length given by uniform randomly is in $L$.
    We then define the \emph{asymptotic probability} of $L$ as $\mu (L) = \lim_{n \to \infty} \mu _n(L)$ if the limit exists.
    We say that $L$ is \emph{almost empty} if $\mu (L) = 0$ and $L$ is \emph{almost full} if $\mu (L) = 1$.
    We say that $L$ obeys \emph{zero-one law} if $L$ is almost empty or almost full.

    In this paper, we now define \emph{$p$-equivalence} by asymptotic probabilities as follows;
    we say that two languages, $L_1$ and $L_2$, are \emph{$p$-equivalent} if $\mu (L_1 \Delta L_2) = 0$.
    $L_1 \simeq_p L_2$ denotes that $L_1$ and $L_2$ are $p$-equivalent and $\alpha_1 \simeq_p \alpha_2$ denotes that $L(\alpha_1) \simeq_p L(\alpha_2)$ for two regular expressions, $\alpha_1$ and $\alpha_2$.
    Note that whether two languages are $p$-equivalent is relative to a given alphabet $A$.
    \begin{example}\label{example mu}
      We first consider a few simple examples about the asymptotic probabilities $\mu$.
      \begin{itemize}
        \item Obviously, $\mu(A^*) = 1$ and $\mu(\emptyset) = 0$.
        \item Let $\alpha_1 = (AA)^*$.
          Then, $\mu_n(L(\alpha_1)) = \begin{cases}
            1 & (\text{if $n$ is even})\\
            0 & (\text{if $n$ is odd})
            \end{cases}$.
          Hence, $\mu(L(\alpha_1))$ does not exist.
        \item
          Let $A = \{a_1,a_2\}$ and $\alpha_2 = a_1^*$.
          Then, $\mu_n(L(\alpha_2)) = \frac{1}{2^n}$.
          Hence, $\mu(L(\alpha_2)) = 0$.
        \item
          Let $A = \{a_1\}$ and $\alpha_3 = a_1^*$.
          Then, $\mu_n(L(\alpha_3)) = 1$.
          Hence, $\mu(L(\alpha_3)) = 1$.
      \end{itemize}
    \end{example}
    \begin{example}\label{example p}
      We now consider a few simple examples about $p$-equivalence.
      \begin{itemize}
        \item Let $A = \{a_1,a_2\}$, $\alpha_1 = A^*$ and $\alpha_1' = a_1 A^*$.
          Then, $\mu_n(L(\alpha_1) \vartriangle L(\alpha_1')) = \frac{|a_2 A^{n-1}|}{|A^n|} = \frac{1}{2}$.\\
          Hence, $\alpha_1 \simeq_p \alpha_1'$ does \emph{not} hold (by that $\mu(L(\alpha_1) \vartriangle L(\alpha_1')) = \frac{1}{2}$).
        \item
          Let $A = \{a_1,a_2,a_3\}$, $\alpha_2 = (a_1 \cup a_2)^*$, and $\alpha_2' = 0$.
          Then, $\mu_n(L(\alpha_2) \vartriangle L(\alpha_2')) = \frac{2^n}{3^n}$.\\
          Hence, $\alpha_2 \simeq_p \alpha_2'$ holds (by that $\mu(L(\alpha_2) \vartriangle L(\alpha_2')) = 0$).
        \item
          Let $A = \{a_1,a_2\}$, $\alpha_3 = (a_1 \cup a_2)^*$, and $\alpha_3' = 0$.
          Then, $\mu_n(L(\alpha_3) \vartriangle L(\alpha_3')) = 1$.\\
          Hence, $\alpha_3 \simeq_p \alpha_3'$ does \emph{not} hold (by that $\mu(L(\alpha_3) \vartriangle L(\alpha_3')) = 1$).
      \end{itemize}
    \end{example}
    \begin{remark}
      The numerator of the definition of $\mu _n(L)$, $|\{s \in A^n \mid s \in L\}|$, is called the \emph{density} of $L$, denoted $d_L(n)$ \cite[Chapter I\hspace{-.1em}X Section 2.2]{pin2010mathematical}.
      In particular, it is said that $L$ has \emph{polynomial density} \cite{szilard1992characterizing} if $d_L(n) = O(n^k)$ for some integer $k > 0$.
      This property is similar to $p$-equivalence.
      Actually, when $|A| \ge 2$, if $L$ has polynomial density, then $\mu (L) = 0$ holds.
      However, these properties are not equivalent because the converse does not clearly hold.
    \end{remark}
    \begin{remark}
      The asymptotic probability over finite strings is like a concrete example of the asymptotic probability over finite $\sigma $-structures.
      Precisely, these are different in that the former is for languages and the latter is for formulas.
      As for regular languages, regular languages are precisely those definable in monadic second-order logic over finite strings ($\mathrm{MSO}[<]$) \cite{buchi1960weak}.
      Thus, the asymptotic probability for regular languages is regarded as a concrete example of the asymptotic probability over finite $\sigma $-structures.
      In additon, the zero-one law considered in this paper is not about ``without order'', but about ``with order''.
      (This difference is important. For example, first-order logic without order ($\mathrm{FO}$) has zero-one law, while first-order logic with order ($\mathrm{FO}[<]$) does not \cite{libkin2013elements}.)
    \end{remark}

  \subsection{Descriptive Complexity}
    In this paper, we use the following results from descriptive complexity.
    \begin{theorem}[{{\cite[Corollary 9.22]{immerman2012descriptive}}}]\label{thm : NL}
      $\mathrm{FO(TC)}$ = NL
    \end{theorem}
    \begin{theorem}[{{\cite[Theorem 9.11]{immerman2012descriptive}}}]\label{thm : L}
      $\mathrm{FO(DTC)}$ = L
    \end{theorem}
    \begin{theorem}[{{\cite[Corollary 10.29]{immerman2012descriptive}}}]\label{thm : PSPACE}
      $\mathrm{SO(TC)}$ = PSPACE
    \end{theorem}
    $\mathrm{TC}$ is a special function such that, for any binary relation $R$, $\mathrm{TC}(R)$ is the transitive closure of $R$.
    $\mathrm{DTC}$ is also a special function such that, for any \emph{deterministic} binary relation $R$ (i.e, $(q,q') \in R \land  (q,q'') \in R \to q' = q''$), $\mathrm{DTC}(R)$ is the transitive closure of $R$.
\section{Fundamental results of $p$-equivalence} \label{section : fund}
  In this section, we give some fundamental results of $p$-equivalence.

  First, $p$-equivalence is an equivalence relation (i.e., $\simeq_p$ is
  \begin{inparaenum}[(1)]
    \item\label{equivalence : 1} reflective : $L_1 \simeq_p L_1$,
    \item\label{equivalence : 2} symmetric : $L_1 \simeq_p L_2 \Rightarrow  L_2 \simeq_p L_1$, and
    \item\label{equivalence : 3} transitive : $L_1 \simeq_p L_2 \land  L_2 \simeq_p L_3 \Rightarrow  L_1 \simeq_p L_3$.
  \end{inparaenum}).
  \ref{equivalence : 1} and \ref{equivalence : 2} obviously hold.
  \ref{equivalence : 3} is proved by the following inequality.
  $0 \le \frac{|(L_1 \vartriangle L_3) \cap A^n|}{|A^n|}
  \le \frac{|(L_1 \vartriangle L_2) \cap A^n|}{|A^n|} +  \frac{|(L_2 \vartriangle L_3) \cap A^n|}{|A^n|} = \mu_n(L_1 \vartriangle L_2) + \mu_n(L_2 \vartriangle L_3)$.
  On the right hand side, by the assumption,
  $\lim_{n \to \infty} \mu _n(L_1 \vartriangle L_2) + \mu _n(L_2 \vartriangle L_3) = 0$.
  Therefore, by the squeeze theorem, $\mu (L_1 \vartriangle L_3) = 0$.
  Hence, $L_1 \simeq_p L_3$.

  \subsection{$p$-equivalence and $f$-equivalence}
    In this subsection, we show a relationship between $p$-equivalence and $f$-equivalence.
    \begin{proposition}\label{proposition : unary finite prob}\leavevmode
      \begin{enumerate}[(1)]
        \item $\operatorname{=} \subseteq \operatorname{\simeq}_f \subseteq \operatorname{\simeq}_p$.
        \item When $|A| \ge 2$, $\operatorname{\simeq}_f \subsetneq \operatorname{\simeq}_p$.
        \item When $|A| = 1$, $\operatorname{\simeq}_f$ is equal to $\operatorname{\simeq}_p$.
      \end{enumerate}
    \end{proposition}
    \begin{proof}\leavevmode
      \begin{inparaenum}[(1)]
        \item $\operatorname{\simeq}_f \subseteq \operatorname{\simeq}_p$ is followed by that, if $L_1 \vartriangle L_2$ is a finite set, then $\mu(L_1 \vartriangle L_2) = 0$.
        \item It is proved by that $\alpha_2 \simeq_p \alpha_2'$ holds, whereas $\alpha_2 \simeq_f \alpha_2'$ does not hold, where $\alpha_2$ and $\alpha_2'$ are the regular expressions in Example \ref{example p}.
        \item We are enough to prove that $\operatorname{\simeq}_f \supseteq \operatorname{\simeq}_p$.
          We prove the contraction , i.e., if $L_1 \not\simeq_f L_2$, then $L_1 \not\simeq_p L_2$.
          Note that $\mu _n(L_1 \vartriangle L_2)$ is $0$ or $1$ because $|A| = 1$ and then $|A^n| = 1$.
          If $L_1 \not\simeq_f L_2$, then $L_1 \vartriangle L_2$ is an infinite set, i.e., $\mu _n(L_1 \vartriangle L_2) = 1$ occurs infinitely.
          Therefore, $\lim_{n \to \infty} \mu _n(L_1 \vartriangle L_2) \neq 0$.
          Hence, $L_1 \not\simeq_p L_2$.
      \end{inparaenum}
    \end{proof}

  \subsection{A robustness of $p$-equivalence}
      We have defined the asymptotic probability of $L$ as (1) $\mu _n(L) = \frac{|\{s \in A^n \mid s \in L\}|}{|A^n|}$.
      However, some other definitions of the asymptotic probability of $L$ have been considered, for example,
      \begin{inparaenum}[(1)]
        \setcounter{enumi}{1}
        \item $\mu ^*_n(L) = \frac{|\{s \in A^{< n} \mid s \in L\}|}{|A^{< n}|}$ and
        \item $\delta_n(L) =  \frac{\sum_{k = 0}^{n-1} \mu _k(L)}{n}$,
      \end{inparaenum}
      where $A^{< n} = \bigcup_{0 \le k < n} A^k$.
      ($\mu _n$ is used by \cite{berstel1973densite}, Salomaa and Soittola \cite{salomaa2012automata}, Sin'ya \cite{DBLP:journals/corr/Sinya15a}, and us;
      $\mu ^*_n$ is used by Berstel \cite{berstel1973densite};
      $\delta_n$ is used by Berstel et al. \cite{berstel2010codes}.
      More details are written in \cite{sinya-phd16}.)
      Let $\mu ^*(L) = \lim_{n \to \infty}\mu ^*_n(L)$ and $\delta(L) = \lim_{n \to \infty}\delta_n(L)$ in the same way as $\mu (L)$.

      Proposition \ref{prop : equiv} says that the three almost equivalences defined by $\mu $, $\mu ^*$, and $\delta$ are all equivalent over regular languages.
      To prove it, we recall the following two theorems.
      \begin{theorem}[Stolz-Ces\`aro theorem (See e.g., \cite{muresan2015concrete})]\label{theorem : Stolz}
        If $\lim_{n \to \infty} \frac{a_{n+1} - a_n}{b_{n+1} - b_n} = l$, then $\lim_{n \to \infty} \frac{a_n}{b_n} = l$, where $\{a_n\}_{n = 0}^\infty$ is a sequence of integers, $\{b_n\}_{n = 0}^\infty$ is a sequence of integers and strictly monotone, and $l$ is a real number.
      \end{theorem}
      \begin{theorem}[Lynch \cite{lynch1993convergence}]\label{theorem : Lynch}
        For any regular language $L$, there exists a positive integer $a$ such that, for any integer $0 \le b < a$, $\lim_{n \to \infty} \mu_{a n + b}(L)$ exists.
        (Let $l_b$ be $\lim_{n \to \infty} \mu_{a n + b}(L)$.)
      \end{theorem}
      \begin{proposition}\label{prop : equiv}
        For any regular language $L$,
        the following three conditions are all equivalent.
        \begin{inparaenum}[(1)]
          \item $\mu (L) = 0$;
          \item $\mu ^*(L) = 0$; and
          \item $\delta(L) = 0$.
        \end{inparaenum}
      \end{proposition}
      \begin{proof}
        1. $\Rightarrow $ 2. and 1. $\Rightarrow $ 3. are proved directly by Theorem \ref{theorem : Stolz}.
        (This part holds even if $L$ is not a regular language.)

        Conversely,
        3. $\Rightarrow $ 1. is proved by the following inequality.
        \begin{align*}
          \delta_n(L)
          = \sum_{k = 0}^{n-1} \frac{\mu _k(L)}{n}
          \ge& \sum_{b = 0}^{a-1} \frac{\sum_{k' = 0}^{m-1} \mu _{a k' + b}(L)}{a m} \times \frac{am}{n}
        \end{align*}
        ,where $m = \lfloor \frac{n}{a} \rfloor$ and $a$ is an integer enjoying the properties stated in Theorem \ref{theorem : Lynch}.
        Then, by Theorem \ref{theorem : Stolz} (Let $a_m = \sum_{k' = 0}^{m-1} \mu _{a k' + b}(L)$ and $b_m = am$), the limit of the above formula as $n$ approaches infinity is $\sum_{b = 0}^{a-1} \frac{l_b}{a}$.
        By $\lim_{n \to \infty} \delta_n(L) = 0$ and the squeeze theorem, $l_b = 0$ for every $b$.
        Hence, $\lim_{n \to \infty} \mu _n(L) = 0$.

        Moreover, 2. $\Rightarrow $ 1. is proved by the following inequality.
        \begin{align*}
          \mu ^*_n(L)
          = \sum_{k = 0}^{n-1} \frac{\mu _k(L) \times |A|^k}{\sum_{k = 0}^{n-1} |A|^k}
          \ge& \sum_{b = 0}^{a - 1} \frac{\sum_{k' = 0}^{m -1} \mu _{a k' + b}(L) \times |A|^{a k' + b}}{\sum_{k' = 0}^{m - 1} |A|^{a k' + b}} \times
          \frac{\sum_{k' = 0}^{m - 1} |A|^{a k' + b}}{\sum_{b' = 0}^{a - 1}\sum_{k' = 0}^{m - 1} |A|^{a k' + b'}}
          \times \frac{\sum_{b' = 0}^{a - 1}\sum_{k' = 0}^{m - 1} |A|^{a k' + b'}}{\sum_{k = 0}^{n - 1} |A|^k}\\
          =& \sum_{b = 0}^{a - 1} \frac{\sum_{k' = 0}^{m -1} \mu _{a k' + b}(L) \times |A|^{a k' + b}}{\sum_{k' = 0}^{m - 1} |A|^{a k' + b}} \times
          \frac{|A|^b}{\sum_{b' = 0}^{a - 1}|A|^{b'}}
          \times \frac{\sum_{b' = 0}^{a - 1}\sum_{k' = 0}^{m - 1} |A|^{a k' + b'}}{\sum_{k = 0}^{n - 1} |A|^k}
        \end{align*}
        ,where $m = \lfloor \frac{n}{a} \rfloor$ and $a$ is an integer enjoying the properties stated in Theorem \ref{theorem : Lynch}.
        Then, by Theorem \ref{theorem : Stolz} (Let $a_m = \sum_{k' = 0}^{m -1} \mu _{a k' + b}(L) \times |A|^{a k' + b}$ and $b_m = \sum_{k' = 0}^{m - 1} |A|^{a k' + b}$), the limit of the above formula as $n$ approaches infinity is $\sum_{b = 0}^{a - 1} l_b \times \frac{|A|^b}{\sum_{b' = 0}^{a-1} |A|^{b'}}$.
        By $\lim_{n \to \infty} \mu ^*_n(L) = 0$ and the squeeze theorem, $l_b = 0$ for every $b$.
        Hence, $\lim_{n \to \infty} \mu _n(L) = 0$.
      \end{proof}

  \subsection{The DFA condition}
    In \cite{DBLP:journals/corr/Sinya15a}, the \emph{zero-one law} regarding the above asymptotic probabilities is introduced and some algebraic characterizations are given.
    We now give the DFA condition, which is different from the characterisations in \cite[Theorem 1]{DBLP:journals/corr/Sinya15a}.
    This condition is very useful to construct the algorithms in the following section.
    (This condition can be proved via \cite[Theorem 1]{DBLP:journals/corr/Sinya15a}.
    However, in this paper, we give a proof more directly and simply.)
    \begin{lemma}\label{lemma : DFA cond}
      For any DFA $\mathcal{A} = (Q,A,\delta,q^0,F)$,
      $$\mu (L(\mathcal{A})) \neq 0 \iff \exists q \in F.( \Reachable(q^0,q) \land  \forall q' \in Q. (\Reachable(q,q') \to \Reachable(q',q)))$$
    \end{lemma}
    $\mu (L(\mathcal{A})) \neq 0$ means that either the limit does not exist, or the limit exists and is not equal to $0$.
    \begin{proof}
      Let $\mu _n(q) = \frac{\{s \in A^n \mid \delta(q^0,s) = q\}}{|A|^n}$ and let $\mu _n(Q') = \sum_{q \in Q'} \mu _n(q)$.
      (Note that $\mu _n(L(\mathcal{A})) = \mu _n(F)$.)
      \begin{description}
        \item[$(\Rightarrow )$]
          We prove the contraposition.
          (i.e., if $\forall q \in F.( \Reachable(q^0,q) \to \exists q' \in Q.( \Reachable(q,q') \land  \lnot \Reachable(q',q)))$, then $\mu (L(\mathcal{A})) = 0$.)

          Let $R_q = \{q' \in Q \mid \Reachable(q',q)\}$.
          Then,
          \begin{align*}
            0 \le \mu _k(F)
            = \sum_{q \in F} \mu _k(q)
            \le \sum_{q \in F} \mu _k(R_q)
            \le& \sum_{q \in F} (1 - \frac{1}{|A|^{|Q|}}) \times \mu _{k-|Q|}(R_q) \tag{1}\\
            \le& \dots\\
            \le& \sum_{q \in F} (1 - \frac{1}{|A|^{|Q|}})^{\lfloor \frac{k}{|Q|} \rfloor} \times \mu _{(k \bmod |Q|)}(R_q) \tag{by using (1) repeatedly}\\
            \le& |F| \times (1 - \frac{1}{|A|^{|Q|}})^{\lfloor \frac{k}{|Q|} \rfloor}
          \end{align*}
          (1) is proved as follows.
          It is enough to prove that, for any $q'' \in R_q$, there exists a string $s'$ such that the length is $|Q|$ and $\delta(q'',s') \not\in R_q$.
          First, there exists a string $s_1'$ such that $\delta(q'',s_1') \not\in R_q$ by the assumption.
          We can assume that the length of $s_1'$ is at most $|Q|$ because the shortest length of string $s_1'$ satisfying $\delta(q'',s_1') \not\in R_q$ is at most $|Q|$.
          Second, $\delta(q'',s_1's_2') \not \in R_q$ for any string $s_2'$ by the definition of $R_q$. 
          Then, $s' = s_1's_2'$ satisfies the above condition by choosing a string $s_2'$ whose length is $|Q| - |s_1'|$.

          Hence, by that $\lim_{k \to \infty} |F| \times (1 - \frac{1}{|A|^{|Q|}})^{\lfloor \frac{k}{|Q|} \rfloor} = 0$ and the squeeze theorem, $\mu (L(\mathcal{A})) = \mu (F) = 0$.
        \item[$(\Leftarrow )$]
          Let $s_0$ be a string such that $\delta(q^0,s_0) = q$ and let $S_q$ be the SCC (Strongly Connected Component) containing $q$.
          Note that $S_q$ is a sink SCC by the assumption ($\forall q' \in Q.(\Reachable(q,q') \to \Reachable(q',q))$).
          Then, by that $S_q$ is a sink SCC, $\mu _k(S_q) \ge \frac{1}{|A|^{|s_0|}}$ for any $k \ge |s_0|$.
          By the pigeon hole principle and that $S_q$ is a sink SCC, for any $k \ge |s_0|$, there exists a state $q' \in S_q$ such that $\mu _k(q') \ge \frac{\mu _k(S_q)}{|S_q|}$.
          Let $s'$ be a string such that $\delta(q',s') = q$ and $|s'| \le |S_q|$ (note that we can reach $q$ from any state $q' \in S_q$ at most $|S_q|$ steps.).
          Then,
          \begin{align*}
            \mu _{k+|s'|}(q)
            \ge& \mu _k(q') \times \frac{1}{|A|^{|s'|}} \tag{by $\delta(q',s') = q$}\\
            \ge& \frac{\mu _k(S_q)}{|S_q|} \times \frac{1}{|A|^{|s'|}}
            \ge (\frac{1}{|A|^{|s_0|}} \times
            \frac{1}{|S_q|}) \times \frac{1}{|A|^{|S_q|}}
            \ge (\frac{1}{|A|^{|Q|}} \times
            \frac{1}{|Q|}) \times \frac{1}{|A|^{|Q|}}
          \end{align*}
          for any $k \ge |s_0|$.
          We can prove that $\mu (L(\mathcal{A})) = 0$ (i.e., $\forall \epsilon > 0.\exists N.\forall n > N. |\mu _n(F)| < \epsilon$) is not true by the above inequality.
          ($\epsilon = \frac{1}{|A|^{|Q|}} \times
          \frac{1}{|Q|} \times \frac{1}{|A|^{|Q|}}$ is a counter example.)
          Therefore, $\mu (L(\mathcal{A})) \neq 0$.
      \end{description}
    \end{proof}

    We now introduce the xor automatons of two DFAs. 
    \begin{definition}
      Let $\mathcal{A}_1 = (Q_1,A,\delta_1,q_1^0,F_1)$ and $\mathcal{A}_2 = (Q_2,A,\delta_2,q_2^0,F_2)$ be DFAs.
      Then, the xor automaton of $\mathcal{A}_1$ and $\mathcal{A}_2$, $\mathcal{A}_1 \oplus \mathcal{A}_2$, is the DFA $(Q_1 \times Q_2,A,\delta',(q_1^0,q_2^0), F')$, where
      \begin{enumerate}[(1)]
        \item $\delta'((q_1,q_2),a) = (\delta_1(q_1,a),\delta_2(q_2,a))$; and
        \item $F' = \{(q_1,q_2) \mid q_1 \in F_1 \text{ xor } q_2 \in F_2\}$.
      \end{enumerate}
    \end{definition}
    Then, the next proposition easily follows.
    \begin{proposition}
      For any DFAs $\mathcal{A}_1$ and $\mathcal{A}_2$,
      $L(\mathcal{A}_1 \oplus \mathcal{A}_2) = L(\mathcal{A}_1) \vartriangle L(\mathcal{A}_2)$.
    \end{proposition}
    Moreover, note that we can construct $\mathcal{A}_1 \oplus \mathcal{A}_2$ from $\mathcal{A}_1$ and $\mathcal{A}_2$ in logarithmic space.

\section{The computational complexity upper bounds of $p$-equivalence problems} \label{section : upperbound}
  In this section, we show the computational complexity upper bounds of $p$-equivalence problems.
  In particular, in terms of the (fully) equivalence problems for REGs, some algorithms have already been developed.
  One approach is to transform two regular expressions into two equivalent NFAs by Meyer and Stockmeyer \cite[Proposition 4.11]{stockmeyer1974complexity}.
  We now give algorithms for the $p$-equivalence problems by using standard results from descriptive complexity \cite{immerman2012descriptive}.
  These algorithms are given by the condition in Lemma \ref{lemma : DFA cond}.
  We prove the next theorem.
  \begin{theorem} \label{theorem : in} \leavevmode
    \begin{enumerate}[1.]
      \item The $p$-equivalence problem for DFAs is in NL.
      \item The $p$-equivalence problem for unary DFAs is in L.
      \item The $p$-equivalence problem for NFAs is in PSPACE.
      \item The $p$-equivalence problem for unary NFAs is in coNP.
    \end{enumerate}
  \end{theorem}

  \begin{proof} \leavevmode
    \begin{enumerate}[1.]
      \item
        We first give a reduction from a DFA to a first-order structure.
        Let $\mathcal{M}^\mathcal{A} = \<Q,\{R_a\}_{a \in A},R_-,q^0,F\>$ be the first-order structure corresponding to a DFA $\mathcal{A} = (Q,A,\delta,q^0,F)$, where
        \begin{inparaenum}[(1)]
          \item $R_a \subseteq Q^2$ is a binary relation such that $(q_1,q_2) \in R_a \iff \delta(q_1,a) = q_2$ for any $a \in A$; and
          \item $R_- \subseteq Q^2$ is a binary relation such that $(q_1,q_2) \in R_- \iff \exists a \in A. (q_1,q_2) \in R_a$.
        \end{inparaenum} 
        (Note that we can construct $\mathcal{M}^{\mathcal{A}}$ from $\mathcal{A}$ in logarithmic space.)

        Let $\mathcal{A}_1 = (Q_1,A,\delta_1,q_1^0,F_1)$ and $\mathcal{A}_2 = (Q_2,A,\delta_2,q_2^0,F_2)$ be two given DFAs.
        Then, the first-order structure $\mathcal{M}^{\mathcal{A}_1 \oplus \mathcal{A}_2}$ can be constructed in logarithmic space.
        The DFA condition in Lemma \ref{lemma : DFA cond}, $\exists q \in F'. \Reachable(q^0,q) \land  \forall q' \in Q_1 \times Q_2. \Reachable(q,q') \to \Reachable(q',q)$, can be written in FO(TC) as $\exists q.(F(q) \land  R_-^*(q^0,q) \land  \forall q'.(R_-^*(q,q') \to R_-^*(q',q)))$, where $R_-^*$\footnote{$R_-^*(q,q')$ denotes $TC(R_-)(q,q') \lor  q = q'$.} is the reflective transitive closure of $R_-$.
        Thus, by NL = FO(TC) (Theorem \ref{thm : NL}), the $p$-equivalence problem for DFAs is in NL.
      \item
        In the case of $|A| = 1$, the sentence written in FO(TC), $\exists q.(F(q) \land  R_-^*(q^0,q) \land  \forall q'.(R_-^*(q,q') \to R_-^*(q',q)))$, is also written in FO(DTC)
        because $R_-$ is deterministic by that $\mathcal{A}_1 \oplus \mathcal{A}_2$ is also unary DFA.
        Therefore, by L = FO(DTC) (Theorem \ref{thm : L}), the $p$-equivalence problem for unary DFAs is in L.
      \item
        Let $\mathcal{A}_1 = (Q_1,A,\delta_1,q_1^0,F_1)$ and $\mathcal{A}_2 = (Q_2,A,\delta_2,q_2^0,F_2)$ be two given NFAs.
        Then, we construct a second-order structure from these NFAs.
        Let $\mathcal{M}^{\mathcal{A}_1 \oplus \mathcal{A}_2} = \<Q_1 \uplus Q_2,\{R_a\}_{a \in A},R_-,Q^0,F'\>$ be the second-order structure, where
        \begin{inparaenum}[(1)]
          \item $R_a \subseteq \wp(Q_1 \uplus Q_2)^2$ is a binary second-order relation such that $(Q',Q'') \in R_a \iff \delta_1(Q' \cap Q_1,a) \cup \delta_2(Q' \cap Q_2,a) = Q''$ for any $a \in A$;
          \item $R_- \subseteq \wp(Q_1 \uplus Q_2)^2$ is a binary second-order relation such that $(Q',Q'') \in R_- \iff \exists a. (Q',Q'') \in R_a$;
          \item $Q^0 = \{q_1^0,q_2^0\}$; and
          \item $F' \subseteq \wp(Q_1 \uplus Q_2)$ is a unary second-order relation such that $Q' \in F' \iff (\exists q_1 \in Q'\cap Q_1. q_1 \in F_1) \text{ xor } (\exists q_2 \in Q'\cap Q_2. q_2 \in F_2)$.
        \end{inparaenum}
        (Note that we can construct $\mathcal{M}^{\mathcal{A}_1 \oplus \mathcal{A}_2}$ from $\mathcal{A}_1$ and $\mathcal{A}_2$ in polynomial space.)
        This structure corresponds to the xor automaton of the two DFAs given by powerset construction of these NFAs.

        Then, the DFA condition in Lemma \ref{lemma : DFA cond} can be written in SO(TC) as $\exists Q.(F(Q) \land  R_-^*(Q^0,Q) \land  \forall Q'.(R_-^*(Q,Q') \to R_-^*(Q',Q)))$, where $R_-^*$ is the reflective transitive closure of $R_-$.
        Therefore, by PSPACE = SO(TC) (Theorem \ref{thm : PSPACE}), the $p$-equivalence problem for NFAs is in PSPACE.
      \item In this case, we give a coNP algorithm for the $p$-equivalence problem directly because it may be easier than using Fagin's Theorem \cite{immerman2012descriptive}.

      Let $A$ be the $n \times n$ adjacency matrix generated from a unary NFA $\mathcal{A} = (\{1,\dots,n\},\{0\},\delta,1,F)$.
      More precisely, $A$ is an adjacency matrix such that
      \begin{inparaenum}[(1)]
        \item $(A)_{i,j} = 1$ if $j \in \delta(i,0)$, and
        \item $(A)_{i,j} = 0$ if $j \not\in \delta(i,0)$.
      \end{inparaenum}
      It is immediate that
      $0^n \in L(\mathcal{A})$ if and only if there exists a number $j \in F$ such that $(A^n)_{1,j} = 1$.
      The following algorithm (Algorithm \ref{algorithm : almost equivalence over unary NFA}) is based on the next lemma.
      \begin{lemma}\label{lemma : unary NFA}
        For any unary NFAs, $\mathcal{A}_1$ and $\mathcal{A}_2$,
        $L(\mathcal{A}_1) \not\simeq_p L(\mathcal{A}_2)$ $\iff$ there exists $n$ such that
        \begin{enumerate}[1.]
          \item $2^{|Q_1| + |Q_2|} \le n < 2^{1 + |Q_1| + |Q_2|}$; and
          \item $0^n \in L(\mathcal{A}_1) \vartriangle L(\mathcal{A}_2)$.
        \end{enumerate}
      \end{lemma}
      \begin{proof}\leavevmode
        Note that $L(\mathcal{A}_1) \simeq_p L(\mathcal{A}_2)$ if and only if $L(\mathcal{A}_1) \simeq_f L(\mathcal{A}_2)$ by that these NFAs are unary NFAs and Proposition \ref{proposition : unary finite prob}.
        Then, it is enough to prove that $L(\mathcal{A}_1) \vartriangle L(\mathcal{A}_2)$ is a infinite set if and only if there exists $n$ such that
        \begin{inparaenum}[(1)]
          \item $2^{|Q_1| + |Q_2|} \le n < 2^{1 + |Q_1| + |Q_2|}$; and
          \item $0^n \in L(\mathcal{A}_1) \vartriangle L(\mathcal{A}_2)$.
        \end{inparaenum}
        Let $v_k = (A_1^k \cdot e_1, A_2^k \cdot e_1)$, where $A_1$ and $A_2$ are the adjacency matrices generated from $\mathcal{A}_1$ and $\mathcal{A}_2$, respectively; and $e_1$ is the unit vector $(1,0,\dots,0)$.
        It is immediate that, for any $k \ge 2^{|Q_1| + |Q_2|}$, $v_k$ occurs infinitely in the sequence $\{v_k\}_{k = 0}^\infty$ because the number of the pattern of $v_k$ is at most $2^{|Q_1| + |Q_2|}$.
        Moreover, for any $v$ occurring infinitely in the sequence $\{v_k\}_{k = 0}^\infty$, there exists $k'$ such that $2^{|Q_1| + |Q_2|} \le k' < 2 \times 2^{|Q_1| + |Q_2|}$ and $v = v_{k'}$ because the period of the sequence $\{v_k\}_{k = 0}^\infty$ is at most $2^{|Q_1| + |Q_2|}$.
        Hence, this Lemma is proved.
      \end{proof}
  
      Then, we give an algorithm (Algorithm \ref{algorithm : almost equivalence over unary NFA}) to search a number $n$ such that satisfies the condition 1 and the condition 2 in Lemma \ref{lemma : unary NFA}.
      Nondeterministically ``guess'' the binary representation of $n$, and test whether there is a path in the adjacency matrix of $A_1$ and $A_2$ of length $n$ to accepting states.
      This idea is based on \cite[Theorem 6.1]{meyer1973word} that states that the equivalence problem for unary NFAs is in coNP.
      The algorithm runs in nondeterministically polynomial time.

      \begin{algorithm}[ht]
        \caption{$p$-equivalence Problem for unary NFA}
        \begin{algorithmic}
          \ENSURE $L(\mathcal{A}_1) \simeq_p  L(\mathcal{A}_2)$? ($True$ or $False$)
          \STATE $(A_1',A_2') \Leftarrow (A_1,A_2)$, where $A_1$ and $A_2$ are the adjacency matrices generated from two unary NFAs, $\mathcal{A}_1$ and $\mathcal{A}_2$, respectively.
          \STATE $d \Leftarrow 1$
          \WHILE{$d < 1 + |Q_1| + |Q_2|$}
            \STATE $(A_1',A_2') \Leftarrow (A_1' \times A_1',A_2' \times A_2')$ or $(A_1',A_2') \Leftarrow (A_1' \times A_1' \times A_1,A_2' \times A_2' \times A_2)$ (nondeterministically)
            \STATE $d \Leftarrow d + 1$
          \ENDWHILE
          \IF{$(\exists j.(A_1')_{1,j} = 1)$ xor $(\exists j.(A_2')_{1,j} = 1)$}
          \STATE return $False$
          \ELSE
          \STATE return $True$
          \ENDIF
        \end{algorithmic}
        \label{algorithm : almost equivalence over unary NFA}
      \end{algorithm}
      In Algorithm \ref{algorithm : almost equivalence over unary NFA}, if any process in the algorithm returns $True$, it is shown that $L(\mathcal{A}_1) \simeq_p L(\mathcal{A}_2)$.
      Otherwise (i.e., if there exists a process such that returns $False$), it is shown that $L(\mathcal{A}_1) \not\simeq_p L(\mathcal{A}_2)$.

      Therefore, the $p$-equivalence problem for unary NFAs is in coNP.
    \end{enumerate}
  \end{proof}

  \subsection{Some generalized equivalence problems}
    We conclude this section with a result for some generalized equivalence problems.
    \begin{corollary}\label{cor : in general}
      Let $x$-equivalence problem be an equivalence problem satisfying that 
      the $x$-equivalence problem for DFAs is logarithmic space reducible 
      to the $\Phi_x$-model-checking problem (i.e, the problem to decide whether $\mathcal{M}$ satisfies $\Phi_x$ for a given model $\mathcal{M}$, where $\Phi_x$ is a first-order sentence with transitive closure).
      Then,
      \begin{enumerate}[1.]
        \item The $x$-equivalence problem for DFAs is in NL.
        \item The $x$-equivalence problem for unary DFAs is in L.
        \item The $x$-equivalence problem for NFAs is in PSPACE.
      \end{enumerate}
    \end{corollary}
    For example, $f$-equivalence \cite{Badr:2008:HO:1428728.1428753,ITA:8238099} and $E$-equivalence \cite{DBLP:conf/dlt/HolzerJ12} satisfy the condition of $x$-equivalence, where $E$ is a finite set. 
    The DFA conditions of these equivalences can be easily written in a first-order sentence with transitive closure.

\section{The computational complexity lower bounds of $p$-equivalence problems} \label{section : lowerbound}
  In this section, we show the computational complexity lower bounds of $p$-equivalence problems.

  \begin{theorem}\label{thm : hard}\leavevmode
    \begin{enumerate}[1.]
      \item The $p$-equivalence problem for DFAs is NL-hard.
      \item The $p$-equivalence problem for unary REGs is coNP-hard.
      \item The $p$-equivalence problem for REGs is PSPACE-hard.
    \end{enumerate}
  \end{theorem}
  
  \begin{proof}
    \begin{enumerate}[1.]
      \item
        We reduce the GAP (Graph Accessibility Problem) to these problems, where $GAP = \{G \mid \text{is an $n \times n$ adjacency matrix that has a path from node $1$ to node $n$}\}$.
        (This proof is based on \cite[Theorem 26]{jones1975space}.)
        Note that GAP is NL-hard \cite{jones1975space}.
        We define the DFA $\mathcal{A}_G = (\{-1,1,\dots,n\},\{1,\dots,n\},\delta,1,$ $\{n\})$, where
        \begin{inparaenum}[(1)]
          \item $\delta(i,j) = j$ if $(i,j)$ is an edge of $G$ and $1 \le i < n$;
          \item $\delta(n,j) = n$; and
          \item $\delta(i,j) = -1$ for all other values of $i,j$.
        \end{inparaenum}
        In this reduction, once you visit at $n$, you will not get out from $n$.
        Then, it is immediate that $G \in GAP \iff L(\mathcal{A}_G) \not\simeq_p \emptyset$ and note that this reduction is in logarithmic space.
        Hence, the $p$-equivalence problem for DFAs is coNL-hard.
        By NL = coNL \cite{immerman1988nondeterministic,szelepcsenyi1988method}, the $p$-equivalence problem is also NL-hard.
      \item
        This part can be solved by the same reduction as \cite[Theorem 6.1]{meyer1973word}.
        This is a reduction from the complement of the equivalence problem to 3-SAT. Note that 3-SAT is a well-known NP-hard problem \cite{cook1971complexity}.
        Let the regular expression $E$ and the $k$-th prime number $p_k$ be the same as \cite[Theorem 6.1]{meyer1973word}.
        Intuitively, a string $0^i$ corresponds to an assignment in 3-SAT whose $k$ th variable is True[False] if and only if $i \equiv 1[0] (\textrm{mod } p_k)$
        and $E$ corresponds to a given formula.
        $0^i \not\in L(E)$ means that the assignment corresponding to $0^i$ satisfies the formula corresponding to $E$.

        Then, we can easily show that $L(E) = A^* \iff L(E) \simeq_f A^*$ because, for any two numbers, $i_1$ and $i_2$, such that $i_1 \equiv i_2 (\textrm{mod } \prod_{k = 1}^n p_k)$, $0^{i_1} \in L(E) \iff 0^{i_2} \in L(E)$ holds.
        Therefore, by Proposition \ref{proposition : unary finite prob}, $L(E) = A^* \iff L(E) \simeq_p A^*$.
        Hence, the $p$-equivalence problem for unary REGs is coNP-hard.
      \item
        It is enough to prove that the $p$-equivalence problem for REGs is NLINSPACE-hard because a language that is CSL-hard (i.e, NLINSPACE-hard) is also PSPACE-hard \cite[Lemma 1.10.(1)]{hunt1976equivalence}.
        The reduction of this proof is based on \cite[Proposition 2.4]{hunt1976equivalence}, which is about that the equivalence problem for REGs is PSPACE-hard.
        Intuitively, in these two reductions, a regular expression $\alpha_M^s$ corresponds to a given nondeterministic linear-space bounded Turing machine $M$ and a given input string $s$ and a string $s' \not\in L(\alpha_M^s)$ corresponds to an accepting sequence of $M$ on input $s$.

        Let $M = (Q,A_M,\delta,q^0,q^a)$ be a nondeterministic linear-space bounded Turing machine and $s = a_1 \dots a_n$ be an input string, where
        \begin{inparaenum}[(1)]
          \item $Q$ is a finite set of states;
          \item $A_M$ is a finite alphabet, where $A_M$ always contains the blank symbol $\bl$;
          \item $\delta : Q \times A_M \to \wp(Q \times A_M \times \{L,R\})$ is a transition function;
          \item $q^0 \in Q$ is the initial state; and
          \item $q^a \in Q$ is the acceptance state.
        \end{inparaenum}
        We also require that once the machine enters its acceptance states, it never leaves it.
        $M$ accepts an input $s$ if the machine can reach an acceptance state $q^a$ from the initial configuration (i.e, the header is at the leftmost position, the state is $q^0$, and the tape is $a_1 \dots a_n$) by finitely transitions.
        Then, we construct the REG $\alpha_M^s = \alpha_1 \cup \alpha_2 \cup \alpha_3$ as follows\footnote{A finite set $\{s_1,\dots,s_n\}$ denotes the regular expression $s_1 \cup \dots \cup s_n$ and $A \setminus c$ denotes $A \setminus \{c\}$.};
        \begin{enumerate}
          \item $A = \{\#\} \cup A_M \cup (Q \times A_M)$,
          \item (input error)
            $\alpha_1 = ((A \setminus \#) \cup \#((A \setminus (q^0,a_1)) \cup (q^0,a_1)((A \setminus a_2) \cup a_2((A \setminus a_3) \cup a_3(\dots))))) A^*$,
          \item (acceptance error) $\alpha_2 = (A \setminus (\bigcup \{q^a\} \times A_M))^*$,
          \item transition error) $\alpha_3 = \bigcup_{c_1,c_2,c_3 \in A} (A \setminus (\bigcup \{q^a\} \times A_M))^* c_1 c_2 c_3 A^{n-2} (A^3 \setminus f_M(c_1,c_2,c_3)) A^*$, and
          \item $f_M : A^3 \to \wp(A^3)$ is the transition function for $M$. Formally, each $f_M(c_1,c_2,c_3)$ is the smallest set that satisfies the following conditions:
            \begin{enumerate}[(i.)]
              \item If $c_1 = (q,a_1)$, $c_2 = a_2$, and $(q',a_1',R) \in \delta(q,a_1)$, then $(a_1',(q',a_2),c_3) \in f_M(c_1,c_2,c_3)$;
              \item If $c_1 = (q,a_1)$ and $(q',a_1',L) \in \delta(q,a_1)$, then $(a_1',c_2,c_3) \in f_M(c_1,c_2,c_3)$;
              \item If $c_2 = (q,a_2)$, $c_3 = a_3$, and $(q',a_2',R) \in \delta(q,a_2)$, then $(c_1,a_2',(q',a_3)) \in f_M(c_1,c_2,c_3)$;
              \item If $c_2 = (q,a_2)$, $c_1 = a_1$, and $(q',a_2',L) \in \delta(q,a_2)$, then $((q',a_1),a_2',c_3) \in f_M(c_1,c_2,c_3)$;
              \item If $c_3 = (q,a_3)$, $c_2 = a_2$, and $(q',a_3',L) \in \delta(q,a_3)$, then $(c_1,(q',a_2),a_3') \in f_M(c_1,c_2,c_3)$;
              \item If $c_3 = (q,a_3)$ and $(q',a_3',R) \in \delta(q,a_3)$, then $(c_1,c_2,a_3') \in f_M(c_1,c_2,c_3)$;
              \item If $c_1 = a_1$, $c_2 = a_2$, and $c_3 = a_3$, then $(c_1,c_2,c_3) \in f_M(c_1,c_2,c_3)$.
            \end{enumerate}
        \end{enumerate}
        Note that the regular expression $\alpha_M^s$ can be constructed in polynomial time.
        Then, we prove the next Lemma. This Lemma gives a relationship between $L(\alpha_M^s)$ and acceptance runs of $M$ on the input $s$.
        \begin{lemma}\label{lemma : correspond}
          For any regular expression $\alpha_M^s$ constructed in the above manner and for any string $s'$, 
          $s' \not\in L(\alpha_M^s)$ if and only if $s'$ is in the form of $$\#(q^0,a_1^0) \dots a_n^0\# \dots \#a_1^i \dots (q^i,a_{k_i}^i) \dots a_n^i\# \dots \# a_1^m \dots (q^m,a_{k_m}^m) c_{m+1} \dots c_l$$, where
          \begin{inparaenum}
            \item $s = a_1^0 \dots a_n^0$;
            \item $q^0$ is the initial state in $M$;
            \item $q^m$ is the acceptance state in $M$; and
            \item for each $i$ ($1 \le i < m$), $\#a_1^i \dots (q^i,a_{k_i}^i) \dots a_n^i$ denotes the $i$ th configuration (i.e., in step $i$, each $j$-th ($1 \le j \le n$) character is $a_j^i$, the state is $q^i$, and the header is at the $k_i$-th position) and this configuration is obtained from the $i-1$ th configuration by a transition.
          \end{inparaenum}
        \end{lemma}
        \begin{proof}\leavevmode
          \begin{description}
            \item[(only if)]
              \begin{inparaenum}
                \item and \item are followed by (input error);
                \item (i.e., $q^a$ occurs in $s'$) is followed by (acceptance error);
                \item is followed by (transition error).
              \end{inparaenum}
            \item[(if)]
              First, $s' \not \in L(\alpha_1)$ is followed by that $s'$ is form of $\#(q^0,a_1^0) \cdots a_n^0 \cdots$.
              Second, $s' \not \in L(\alpha_2)$ is followed by that $q^a$ occurs in $s'$.
              Third, $s' \not \in L(\alpha_3)$ is followed by that $s'$ represents valid configurations until $q^a$ does not occur in $s'$.
              Therefore, $s' \not \in L(\alpha_M^s)$.
          \end{description}
        \end{proof}
        It is immediate that any $s'$ satisfying the conditions in Lemma \ref{lemma : correspond} corresponds to an acceptance run of $M$ on the input $s$; and, for any acceptance run of $M$ on the input $s$, there exists a string $s'$ such that satisfies the conditions in Lemma \ref{lemma : correspond}.
        Then, we can prove the next Lemma.
        \begin{lemma}\label{lemma : regular expression and turing machine}
          For any nondeterministic linear-space bounded Turing machine $M$ and for any string $s$,
          the following three conditions are equivalent.
          \begin{enumerate}
            \item $M$ does not accept the input $s$.
            \item $L(\alpha_M^s) = A^*$.
            \item $L(\alpha_M^s) \simeq_p A^*$.
          \end{enumerate}
        \end{lemma}
        \begin{proof}\leavevmode
          (a) $\Leftrightarrow$ (b) is followed by Lemma \ref{lemma : correspond} and the above consideration.
          (b) $\Rightarrow $ (c) is easily followed by $\operatorname{=} \subseteq \operatorname{\simeq}_p$.
          We only prove (c) $\Rightarrow $ (b). We prove the contraposition.
         
          When $L(\alpha_M^s) \neq A^*$, let $s'$ be a string not in $L(\alpha_M^s)$.
          It is immediate that, for any string $s''$, $s's''$ is also in the form of $\#(q^0,a_1^0) \dots a_n^0\# \dots \#a_1^i \dots (q^i,a_{k_i}^i) \dots a_n^i\# \dots \# a_1^m \dots (q^m,a_{k_m}^m) c_{m+1} \dots c_l$.
          (Note that any string matches $c_{m+1} \dots c_l$.)
          
          Therefore, $\mu_{n'}(L(\alpha_M^s)) \le 1 - \frac{1}{|A|^{|s'|}}$ and $\mu_{n'}(L(\alpha_M^s) \vartriangle A^*) = 1 - \mu_{n'}(L(\alpha_M^s)) \ge 1 - (1 - \frac{1}{|A|^{|s'|}}) = \frac{1}{|A|^{|s'|}}$ hold, where $n' \ge |s'|$.
          Hence, by $\mu_{n'}(L(\alpha_M^s) \vartriangle A^*) \neq 0$, $L(\alpha_M^s) \not\simeq_p A^*$.
        \end{proof}

        Thus, we can reduce the membership problem for nondeterministic linear-space bounded Turing machine to the $p$-equivalence problem for REGs.
        Therefore, the $p$-equivalence problem for REGs is PSPACE-hard.
    \end{enumerate}
  \end{proof}

  \begin{remark}
    The principal difference between this reduction and the reduction of \cite[Proposition 2.4]{hunt1976equivalence} is only (transition error).
    By this modification, $L(\alpha) \simeq_p A^* \iff L(\alpha) = A^*$ holds.
  \end{remark}

  The next theorem is obtained from Theorem \ref{theorem : in} and Theorem \ref{thm : hard}.
  \begin{theorem}\label{thm : complete}\leavevmode
    \begin{enumerate}
      \item The $p$-equivalence problem for DFAs is NL-complete.
      \item The $p$-equivalence problem for unary DFAs is in L.
      \item The $p$-equivalence problems for NFAs and REGs are PSPACE-complete.
      \item The $p$-equivalence problems for unary NFAs and unary REGs are coNP-complete.
    \end{enumerate}
  \end{theorem}
  \begin{proof}
     We can transform any regular expression $\alpha$ into an NFA $\mathcal{A}_\alpha$ such that $L(\alpha) = L(\mathcal{A}_\alpha)$ in polynomial time (e.g., Thompson's construction \cite{Thompson:1968:PTR:363347.363387,sakarovitch2009elements}).
     For example, it is an easy consequence that the $p$-equivalence problem for REGs is in PSPACE by the construction and Theorem \ref{theorem : in}.
     It is also an easy consequence that the $p$-equivalence problem for NFAs is PSPACE-hard by the construction and Theorem \ref{thm : hard}.
  \end{proof}

  \subsection{Some generalized equivalence problems}
    We conclude this section with a result for some generalized equivalence problems.
    \begin{corollary}\label{corollary : gen-hard}
      Let $x$-equivalence problem be an equivalence problem satisfying that $\operatorname{=} \subseteq \operatorname{\simeq}_x \subseteq \operatorname{\simeq}_p$.
      Then, 
      \begin{enumerate}[(1)]
        \item The $x$-equivalence problems for REGs and NFAs are PSPACE-hard.
        \item The $x$-equivalence problem for DFAs is NL-hard.
        \item The $x$-equivalence problems for unary REGs and unary NFAs are coNP-hard.
      \end{enumerate}
    \end{corollary}
    \begin{proof}
      We first show that $L(\alpha_M^s) \simeq_x A^* \iff L(\alpha_M^s) \simeq_p A^*$.
      \begin{description}
        \item[$(\Rightarrow )$]
          It is followed by that $\operatorname{\simeq}_x \subseteq \operatorname{\simeq}_p$.
        \item[$(\Leftarrow )$]
          By $L(\alpha_M^s) = A^* \iff L(\alpha_M^s) \simeq_p A^*$ (Lemma \ref{lemma : regular expression and turing machine}), $L(\alpha_M^s) = A^*$.
          Then, $L(\alpha) \simeq_x A^*$ is followed by $\operatorname{=} \subseteq \operatorname{\simeq}_x$.
      \end{description}
      Therefore, we can reduce the membership problem for nondeterministic linear-space bounded Turing machine to the $x$-equivalence problem for REGs by using the same reduction in Theorem \ref{thm : hard}.
      Hence, (1) is proved.

      (2) and (3) are also proved in the same way as (1).
      (2) is followed by that $L(\mathcal{A}_G) \not\simeq_p \emptyset \iff L(\mathcal{A}_G) \neq \emptyset$ is described in Theorem \ref{thm : hard}.
      (3) is followed by that $L(E) \simeq_p A^* \iff L(E) = A^*$ is described in Theorem \ref{thm : hard}.
    \end{proof}

    Moreover, the next corollary is obtained from Corollary \ref{cor : in general} and Corollary \ref{corollary : gen-hard}
    \begin{corollary}\label{corollary : complete}
      Let $x$-equivalence problem be an equivalence problem satisfying that 
      \begin{inparaenum}[(1)]
        \item the $x$-equivalence problem for DFAs is logarithmic space reducible 
          to the $\Phi_x$-model-checking problem; and
        \item $\operatorname{=} \subseteq \operatorname{\simeq}_x \subseteq \operatorname{\simeq}_p$.
      \end{inparaenum}
      Then, 
      \begin{enumerate}[(1)]
        \item The $x$-equivalence problems for REGs and NFAs are PSPACE-complete.
        \item The $x$-equivalence problem for DFAs is NL-complete.
      \end{enumerate}
    \end{corollary}

    For example, $f$-equivalence and $E$-equivalence satisfy the condition of $x$-equivalence, where $E$ is a finite set.
    Hence, for any finite set $E$, the $E$-equivalence problem for NFAs \cite{DBLP:conf/dlt/HolzerJ12} is also PSPACE-complete, whereas $E$ is fixed.
  \section{The computational complexities of zero-one law} \label{section : zero-one law}
    We define the \emph{zero-one problem} as the problem to decide whether a given language $L$ obeys zero-one law \cite{DBLP:journals/corr/Sinya15a} (i.e., $\mu (L) = 0$ or $\mu (L) = 1$).
    (In terms of time complexity, the zero-one problem for DFA is $O(|A| n)$ \cite{DBLP:journals/corr/Sinya15a}, where $|A|$ is the size of alphabet and $n$ is the number of states.)
    
    In this section, we show that the zero-one problem and the $p$-equivalence problem are the same in terms of the computational complexities.
    \begin{corollary}\label{corollary : zero-one}\leavevmode
      \begin{enumerate}
        \item The zero-one problem for REG and NFA are PSPACE-complete.
        \item The zero-one problem for DFA is NL-complete.
        \item The zero-one problem for unary REG and unary NFA are coNP-complete.
        \item The zero-one problem for unary DFA is in L.
      \end{enumerate}
    \end{corollary}
    \begin{proof}
      First, each zero-one problem can be solved by two $p$-equivalence problems as $L \simeq_p \emptyset \lor  L \simeq_p A^*$.
      Therefore, the zero-one problems are not harder than $p$-equivalence problems.
      For example, if $p$-equivalence problem for REGs is in PSPACE, then zero-one problem for REG is also in PSPACE.

      It is also proved that the computational hardness of the zero-one problems are given in the almost same way as the computational hardness for the $p$-equivalence problems as follows.
      \begin{description}
        \item[REG and NFA] In Theorem \ref{thm : hard}, for any regular expression $\alpha_M^s$ constructed from $M$ and $s$,
        $L(\alpha_M^s) \not \simeq_p \emptyset$ is easily followed by that $L(\#\# A^*) \subseteq L(\alpha_M^s)$.
        Therefore, $L(\alpha_M^s)$ has zero-one law $\iff$ $L(\alpha_M^s) \simeq_p A^*$.
        \item[DFA] In Theorem \ref{thm : hard}, we intentionally create a path to $0$ by a new character $e$.
        More precisely, we define the DFA $\mathcal{A}_G = (\{0,1,\dots,n\},\{e,1,\dots,n\},\delta,1,\{n\})$, where
        \begin{inparaenum}[(1)]
          \item if $i \neq n$, then $\delta(i,e) = 0$;
          \item if $i = n$, then $\delta(i,e) = n$; and
          \item otherwise, $\delta(i,j)$ is the same as $\delta(i,j)$ in Theorem \ref{thm : hard}.
        \end{inparaenum}
        Then, $L(\mathcal{A}_G) \not\simeq_p A^*$ is easily followed by that, for any string $s \in L(e A^*)$, $s \not\in L(\mathcal{A}_G)$.
        Therefore, $L(\mathcal{A}_G)$ has zero-one law $\iff$ $L(\mathcal{A}_G) \simeq_p \emptyset$.
        \item[unary REG and unary NFA] We can use the reduction in \cite[Theorem 6.1]{meyer1973word}.
        In \cite[Theorem 6.1.]{meyer1973word}, $E$ is always an infinite set.
        Therefore, $L(E) \not \simeq_f \emptyset$.
        By Lemma \ref{lemma : unary NFA}, $L(E) \not \simeq_p \emptyset$.
        Hence, $E$ has zero-one law $\iff$ $L(E) \simeq_p A^*$.
      \end{description}
    \end{proof}
\section{Conclusion and Future Work} \label{section : conclusion}
  We have got the following results (Table \ref{table : result}).
  In regular languages, the $p$-equivalence problems and the (fully) equivalence problems are the same in terms of the computational complexities.
  Moreover, we have got the same complexity computational results for some generalized equivalence problems.

  One of the possible future works is to study about $p$-equivalence for more complex language classes (e.g., context free languages).
  In connection with almost-equivalence, it is also interesting to characterize hyper-minimization based on $p$-equivalence like \cite[Theorem 3.4]{ITA:8238099}.
  \begin{table}[h]
    \begin{tabular}{|l||c|c|c|c|c|c|}
      \hline & \multicolumn{3}{|c|}{unary alphabet ($|A| = 1$)}  & \multicolumn{3}{|c|}{general case}\\
      \hline & REG & DFA & NFA & REG & DFA & NFA\\
      \hline\hline equivalence & coNP-c & in L & coNP-c & PSPACE-c & NL-c & PSPACE-c\\
      & \cite{meyer1973word} & \cite{jones1975space} & \cite{meyer1973word} & \cite{meyer1973word} & \cite{jones1975space} & \cite{meyer1973word}\\
      \hline $p$-equivalence & coNP-c& in L & coNP-c & PSPACE-c & NL-c & PSPACE-c\\
      & (Th.\ref{thm : complete}) & (Th.\ref{theorem : in}) & (Th.\ref{thm : complete}) & (Th.\ref{thm : complete}) & (Th.\ref{thm : complete}) & (Th.\ref{thm : complete})\\
      \hline zero-one & coNP-c & in L & coNP-c & PSPACE-c & NL-c & PSPACE-c\\
      & (Cor.\ref{corollary : zero-one}) & (Cor.\ref{corollary : zero-one}) & (Cor.\ref{corollary : zero-one}) & (Cor.\ref{corollary : zero-one}) & (Cor.\ref{corollary : zero-one}) & (Cor.\ref{corollary : zero-one})\\
      \hline
    \end{tabular}
    \caption{The computational complexities of some problems for regular languages}
    \label{table : result}
  \end{table}
  \vspace{-5ex}
\section{Acknowledgements}
  I would like to thank Ryoma Sin'ya for suggesting holding Proposition \ref{prop : equiv} and for giving some beneficial comments.
  This work was supported by JSPS KAKENHI Grant Number 16J08119.

\bibliographystyle{eptcs}
\bibliography{gandalf2016}
\end{document}